\documentclass{article}

\usepackage{iclr2021_conference}
\usepackage{times}
\usepackage[utf8]{inputenc}
\usepackage{url}
\usepackage{amsmath}
\usepackage{amssymb}
\usepackage{amsthm}
\usepackage{graphicx}
\usepackage{subfigure}
\usepackage{booktabs}
\usepackage{hyperref}
\usepackage{algorithm}
\usepackage{algpseudocode}
\usepackage{color}
\usepackage{bm}

\newtheorem{theorem}{Theorem}

\theoremstyle{definition}
\newtheorem{definition}{Definition}

\usepackage{xcolor}

\newcommand{\x}{{\mathbf{x}}}

\newcommand{\X}{{\mathbf{X}}}
\newcommand{\hatX}{{\mathbf{\hat X}}}

\title{A coding theorem for the\\ rate--distortion--perception function}
\author{}
\author{Lucas Theis \\
Google Research \\
\texttt{theis@google.com}
\And 
Aaron B. Wagner \\
Cornell University\\
\texttt{wagner@cornell.edu}}

\iclrfinalcopy

\begin{document}

\maketitle

\begin{abstract}
    The \textit{rate--distortion--perception function} \citep[RDPF;][]{blau2019rethinking} has emerged as a useful tool for thinking about realism and distortion of reconstructions in lossy compression. Unlike the rate--distortion function, however, it is unknown whether encoders and decoders exist that achieve the rate suggested by the RDPF. Building on results by \cite{li2017}, we show that the RDPF can indeed be achieved using stochastic, variable-length codes. For this class of codes, we also prove that the RDPF lower-bounds the achievable rate.
\end{abstract}

\section{Introduction}

Lossy compression seeks to represent data with as few bits as possible while also reconstructing the data as closely as possible. In most applications of compression, however, low bit-rates and low distortion of the input are not the only desiderata. Another criterion often considered is the realism or ``perceptual quality'' of reconstructions. While distortion is typically measured using semimetrics, realism is more readily expressed in terms of a divergence between the source and reconstruction ensembles. When this divergence is zero, reconstructions are indistinguishable from real data and therefore maximally realistic, even though the distortion between source-reconstruction pairs may be large.

Approximately optimizing divergences using adversarial networks \citep{goodfellow2014gan} has been key to many approaches in generative modeling. In the context of compression,
the joint optimization of distortions and divergences has shown particular promise for photo-realistic reconstructions at low bit-rates \cite[e.g.,][]{rippel2017waveone,ledig2017srgan,santurkar2018generative,agustsson2019extreme}. \cite{blau2018tradeoff} provided a theoretical motivation for this approach by showing that in general there is a tension between divergences and distortions, suggesting that optimizing distortion alone is not sufficient to achieve realism.

\cite{blau2019rethinking} further formalized the idea of a trade-off between different desiderata and introduced the \textit{(information) rate--distortion--perception function (RDPF)}. This function gives the hypothetical rate of an optimal code under constraints on distortion and realism. However, unlike the conventional rate--distortion function \citep{shannon1948,shannon1959}, it is not yet clear whether the rate suggested by the RDPF is operationally achievable, and under what conditions.

We show that the RDPF indeed describes the limit of compression under suitable constraints on distortion and realism. In Section~\ref{sec:preliminaries} we will define the most important concepts needed for our results in Section~\ref{sec:achievability}. 

\section{Preliminaries}
\label{sec:preliminaries}

In practice, the data source $\X$ is generally modeled as a real-valued (though possibly
discrete) random variable
$\X: \Omega \rightarrow \mathbb{R}^M$. Our results 
apply to sources in more abstract spaces, however.
Logarithms are to the base 2 and entropy and mutual information are measured in bits.

\vspace{7pt}
\begin{definition}
    For any distortion $d$, divergence $D$, and a random variable $\X$, the \textit{(information) rate--distortion--perception function} (RDPF) is defined as \citep{blau2019rethinking}
    \begin{align}
        R(\theta_d, \theta_D) = \inf_{P_{\hatX \mid \X}}\ I[\X, \hatX]
        \quad\text{s.t.}\quad \mathbb{E}[d(\X, \hatX)] \leq \theta_d
        \quad\text{and}\quad D[P_\X, P_\hatX] \leq \theta_D.
    \end{align}
\end{definition}

More generally, we can define a rate function for an arbitrary set of constraints on $P_{\X, \hatX}$ as follows.

\vspace{7pt}
\begin{definition}
    For a source $\X \sim P_{\X}$ and a set of real-valued functions $D_i$ of joint distributions $P_{\X, \hatX}$, the \textit{(information) rate function} (IRF) is defined as
    \begin{align}
        R(\bm{\theta}) = \inf_{P_{\hatX \mid \X}}\ I[\X, \hatX]
        \quad\text{s.t.}\quad \forall i: D_i[ P_{\X, \hatX} ] \leq \theta_i.
        \label{eq:IRF}
    \end{align}
\end{definition}

The RDPF is a special case of the IRF where we chose
\begin{align}
    D_1[P_{\X, \hatX}] = \mathbb{E}[d(\X, \hatX)]
    \quad\text{and}\quad
    D_2[P_{\X, \hatX}] = D[P_{\X}, P_{\hatX}].
\end{align}
The generalized constraint in Equation~\ref{eq:IRF} was anticipated by \cite{shannon1948}.
We prove our results for the IRF, from which results on the RDPF immediately
follow. The ready unification of distortion and perception suggests that these
concepts are less distinct than it might first appear. Note that rate-distortion theorems
with multiple, competing constraints have been considered in other contexts~\cite[][Exercise 2.2.14]{kakavand2006,Csiszar1981}.

So far we have only defined informational rate functions which may or may not have any connection to the operational compression problem. The hope is, of course, that the value of these functions represents the rate of the best codes whose reconstructions satisfy the given constraints. That is, we hope that a code exists which \textit{achieves} the given rate and, conversely, that no code can do better.

\vspace{7pt}
\begin{definition}
    For an arbitrary set $\mathcal{X}$,  we define a \textit{(stochastic) encoder} as any function contained in
    \begin{align}
        \mathcal{F} = \{ f: \mathcal{X} \times \mathbb{R} \rightarrow \mathbb{N}_0 \}.
    \end{align}
    Similarly, we define \textit{(stochastic) decoders} as elements of 
    \begin{align}
        \mathcal{G} = \{ g: \mathbb{N}_0 \times \mathbb{R} \rightarrow \mathcal{X} \}.
    \end{align}
    A \textit{(stochastic) code} is an element of $\mathcal{F} \times \mathcal{G}$.
\end{definition}

Note that a stochastic encoder expects data and an additional source of randomness as input. Without loss of generality, we can assume that this randomness is represented by a single real number (or an infinite number of bits).
Similar to the encoder, the decoder receives additional random bits as input. The encoder and decoder may receive the same random bits or they may each have their own source of randomness, and this has implications for the efficiency of a code \citep[e.g.,][]{cuff2008}. Throughout the paper, we
assume that the encoder and decoder receive the same bits.

Next we define our notion of achievability. Following the conventions in information theory, we distinguish between one-shot achievability and (asymptotic) achievability.

\vspace{7pt}
\begin{definition}
   For a source $\X \sim P_\X$ and a given set of constraints, we say that a rate $R$ is \textit{one-shot achievable} if an  encoder $f \in \mathcal{F}$, a decoder $g \in \mathcal{G}$, and a random variable $U$ exist with
    \begin{align}
        K = f(\X, U)
        \quad\text{and}\quad
        \hatX = g(K, U)
    \end{align}
    such that the joint distribution $P_{\X, \hatX}$ satisfies the constraints and the conditional entropy of $K$ is not more than $R$, $H[K \mid U] \leq R$.
\end{definition}

This definition is justified by the fact that a variable-rate prefix-free entropy code always exists whose coding cost is within $1$ bit of the entropy, $H[K \mid U]$. Note that an entropy coder may use the variable $U$ to assign bits to $K$ since $U$ is known to both the encoder and the decoder. 

In one-shot achievability, the source $\X$ is not assumed to have any discernible structure. This is
in contrast to (asymptotic) achievability, for which the source is assumed to be a long i.i.d.\ sequence.

\vspace{7pt}
\begin{definition}
    For a source $\X \sim P_\X$ and a given set of constraints, we say that a rate $R$ is \textit{(asymptotically) achievable} if there exists a sequence of
    encoders
    $f_N: \mathcal{X}^N \times \mathbb{R} \rightarrow \mathbb{N}_0$ and decoders
    $g_N: \mathbb{N}_0 \times \mathbb{R} \rightarrow \mathcal{X}^N$ with
    \begin{align}
        K_N = f(\X^N, U)
        \quad\text{and}\quad
        \hatX^N = g(K_N, U)
    \end{align}
    such that each joint distribution $P_{\X_n, \hatX_n}$ ($n = 1, \dots, N$) satisfies the constraints and
    \begin{align}
        \lim_{N \rightarrow \infty} H[K_N \mid U]/N \leq R.
        \label{eq:asymptoticrate}
    \end{align}
    \label{def:asymptotic}
\end{definition}
An important detail is that this definition requires that the reconstruction of each individual data point satisfies the constraints, not just when averaged over the $N$ points.

\section{Achievability of the IRF (and RDPF)}
\label{sec:achievability}

We first prove the one-shot achievability of an upper bound on the IRF. 

\vspace{7pt}
\begin{theorem}
    \label{th:oneshot}
    Let an arbitrary source $\X \sim P_\X$ and constraints $D_i[P_{\X, \hatX}] \leq \theta_i$ be given. If
    \begin{align}
        R > R(\bm{\theta}) + \log(R(\bm{\theta}) + 1) + 4,
    \end{align}
    then $R$ is one-shot achievable.
\end{theorem}
\begin{proof}
    This result is a direct consequence of the properties of the \textit{Poisson functional representation} introduced by \cite{li2017}.
    By the definition of the IRF, there exists a $P_{\hatX \mid \X}$ such that
    \begin{align}
        \forall i: D_i[ P_{\X, \hatX}] \leq \theta_i
        \quad\text{and}\quad
        I[\X, \hatX] \leq R(\bm{\theta}) + \varepsilon
    \end{align}
    for any $\varepsilon > 0$. Using a shared source of randomness, an encoder and decoder are both able to generate a draw from the following random variables for $i \in \mathbb{N}$,
    \begin{align}
        \label{eq:rvs}
        \hatX_i \sim P_{\hatX}, \quad S_i \sim \text{Exp}(1), \quad T_i &= \textstyle\sum_{j = 1}^i S_i.
    \end{align}
    Here, $S_i$ are exponentially distributed so that $T_i$ are the epochs of a Poisson process. For an input $\x$, the encoder selects one of the candidates $\hatX_i$ as follows,
    \begin{align}
        K &= \underset{i \in \mathbb{N}}{\text{argmin}}\ T_i \frac{d P_\hatX}{dP_{\hatX \mid \X}(\,\cdot \mid \x)}(\hatX_i).
    \end{align}
    After receiving $K$, the decoder reconstructs $\hatX_K$. That is,
    \begin{align}
        f(\x, U) = K \quad \text{and} \quad g(K, U) = \hatX_K,
    \end{align}
    where $U$ is a random variable representing the shared source of randomness used to construct the random variables in Equation~\ref{eq:rvs}. \cite{li2017} proved that $\hatX_K \sim P_{\hatX \mid \X}$, that is, the code is communicating a sample from the conditional distribution. It was further shown that
    \begin{align}
        H[K] < I[\X, \hatX] + \log(I[\X, \hatX] + 1) + 4,
    \end{align}
    that is, the indices' entropy is not much more than the mutual information. Hence,
    \begin{align}
        H[K \mid U] \leq H[K] < R(\bm{\theta}) + \log(R(\bm{\theta}) + \varepsilon) + 4 + \varepsilon < R,
    \end{align}
    if we choose $\varepsilon$ small enough.
\end{proof}

Next, we show that the IRF is a lower bound on the one-shot achievable rate. This extends a similar result by \citet[][Appendix C]{blau2019rethinking} from deterministic to stochastic codes.

\vspace{7pt}
\begin{theorem}
    Let an arbitrary source $\X \sim P_\X$ and constraints $D_i[P_{\X, \hatX}] \leq \theta_i$ be given.
    If $R < \infty$ is one-shot achievable,
    then $R \ge R(\bm{\theta})$.
    \label{theorem:oneshotlower}
\end{theorem}
\begin{proof}
    Applying properties
    of mutual information and discrete entropy, we have
     \begin{align}
        H[f(\X,U) \mid U] & \ge I[\X, f(\X,U) \mid U]  \\
                       & \ge I[\X, \hat{\X} \mid U] \\
                       & = I[\X, \hat{\X} \mid U] + I[\X, U] \\
                       & = I[\X, (\hat{\X},U)] 
                        \ge  I[\X,\hat{\X}] \ge R(\bm{\theta}). \qedhere
   \end{align}
\end{proof}
We are now ready to prove the IRF (and hence RDPF) coding theorem.
\vspace{4pt}
\begin{theorem}
    \label{th:asymptotic}
    Let an arbitrary source $\X \sim P_\X$ and constraints $D_i[P_{\X, \hatX}] \leq \theta_i$ be given. Then $R < \infty$ is achievable if and only if 
    $R \ge R(\bm{\theta})$.
\end{theorem}
\begin{proof}
    We first show that if $R(\bm{\theta}) < \infty$, then $R(\bm{\theta})$ is itself achievable. 
    By the definition of the IRF, for each $N > 0$ there exists a (potentially different) $P_{\hatX \mid \X}$ such that
    \begin{align}
        I[\X, \hatX] <  R(\bm{\theta}) + \frac{1}{N}
    \end{align}
    and $P_{\X, \hatX}$ satisfies the constraints. Following the same approach as in Theorem~\ref{th:oneshot}, we can construct a code communicating a sample from
    \begin{align}
        P_{\hatX^N \mid \X^N} = \textstyle\prod_{n = 1}^N P_{\hatX \mid \X},
    \end{align}
    where the right-hand side indicates a product measure, that is, $\hatX_n$ will only depend on $\X_n$. Each individual data point and reconstruction $\smash{(\X_n, \hatX_n)}$ follows the marginal distribution $P_{\X, \hatX}$ and thus satisfies the constraints. We encode the data points jointly into some $K_N$ so that its average entropy is now at most
    \begin{align}
        \frac{H[K_N]}{N}
        &< \frac{1}{N} (I[\X^N, \hatX^N] + \log(I[\X^N, \hatX^N] + 1) + 4) \\
        &= I[\X, \hatX] + \frac{1}{N} \log(N I[\X, \hatX] + 1) + \frac{4}{N} \\
        &< R(\bm{\theta}) + \frac{1}{N}\log(N R(\bm{\theta}) + 2) + \frac{5}{N}.
    \end{align}
    This converges to $R(\bm{\theta})$ in the limit of large $N$, proving the former's achievability.
    Turning to the converse, for each $N$, let $R_N(\bm{\theta})$ denote the IRF of $\X^N$ with constraints $D_i[P_{\X_n, \hatX_n}] \leq \theta_i\ \forall\ i, n$. If $R$ is achievable then there exists a sequence of codes $(f_N, g_N)$ satisfying these constraints and
    \begin{align}
        \lim_{N \rightarrow \infty} H[K_N \mid U]/N \leq R,
    \end{align}
    where $K_N = f_N(\X_N, U)$. But by Theorem~\ref{theorem:oneshotlower} we must have $H[K_N \mid U] \geq R_N(\bm{\theta})$ and so
    \begin{equation}
        R \ge \lim_{N \rightarrow \infty} H[K_N \mid U]/N \ge \lim_{N \rightarrow \infty} R_N(\bm{\theta})/N.
    \end{equation}
    It thus suffices to show that $R_N(\bm{\theta}) \ge N R(\bm{\theta})$. But this follows from the chain rule for mutual information, the data processing theorem, and the independence of $X_n$ from $X_m$ ($m \neq n$):
    \begin{align*}
        I[\X^N,\hat{\X}^N] 
        &= \textstyle\sum_{n = 1}^N I[\X_n,\hat{\X}^N \mid \X_1,\ldots,\X_{n-1}] \\
        &= \textstyle\sum_{n = 1}^N I[\X_n,\hat{\X}^N \mid \X_1,\ldots,\X_{n-1}] 
            + \sum_{n = 1}^N I[\X_n,(\X_1,\ldots,\X_{n-1})] \\
        &= \textstyle\sum_{n = 1}^N I[\X_n,(\hat{\X}^N, \X_1,\ldots,\X_{n-1})] \\
        &\ge \textstyle\sum_{n = 1}^N I[\X_n,\hat{\X}_n] 
        \ge \textstyle\sum_{n = 1}^N R(\bm{\theta}) = N R(\bm{\theta}). \qedhere
    \end{align*}
\end{proof}

\section{Discussion}

By building on the results of \cite{li2017} we proved a complete coding theorem for the
RDPF (and more generally, any IRF). This clarifies the practical relevance of the RDPF. Note that our asymptotic formulation in Definition~\ref{def:asymptotic}
requires the constraints to hold for each $n$, not just on average. This stronger constraint
is satisfied by the construction of \cite{li2017} and eliminated the need for the assumption, used
by~\cite{blau2019rethinking}, that the divergence measure is convex. On the other hand, unlike
\cite{blau2019rethinking},
our formulation allowed for variable-length coding and common randomness between the encoder and decoder. The case of fixed-rate and/or deterministic codes is also of interest but is left for future research.

\subsubsection*{Acknowledgments}
We would like to thank Johannes Ballé, Eirikur Agustsson and Ashok Popat for helpful discussions and feedback on the manuscript. The second author was supported
by the US National Science Foundation under grants CCF-1617673, CCF-1934985, and
CCF-2008266 and the US Army Research Office under grant W911NF-18-1-0426.

\bibliographystyle{plainnat}
\bibliography{references}

\end{document}